\documentclass{article}
\usepackage{ijcai18}

\usepackage{times}
\usepackage{xcolor}
\usepackage{soul}
\usepackage[utf8]{inputenc}
\usepackage[small]{caption}

\usepackage[ruled,vlined,oldcommands,linesnumbered,algo2e]{algorithm2e}

\usepackage{graphicx}

\usepackage{xspace,comment,multirow,url,amsmath,amsthm}

\newtheorem{theorem}{Theorem}
\newtheorem{mydef}{Definition}
\usepackage{enumitem}
\usepackage{booktabs}

\usepackage{setspace}
\usepackage{amssymb}
\usepackage{color}
\usepackage{flexisym}
\usepackage{balance}
\usepackage{array}
\usepackage{tabulary}

\usepackage{etoolbox}
\patchcmd{\maketitle}{\@copyrightpermission}{}{}{}

\DeclareSymbolFont{matha}{OML}{txmi}{m}{it}%
\DeclareMathSymbol{\q}{\mathord}{matha}{113}
\DeclareMathSymbol{\w}{\mathord}{matha}{119}

\def\cZ{{\mathcal Z}}
\relax

\let\hat\widehat

\def\E{{\mathbb E}}
\def\Pr{{\mathbb P}}

\def\cC{{\cal C}}

\newcommand{\be}{\begin{equation}}
\newcommand{\ee}{\end{equation}}
\newcommand{\bea}{\begin{eqnarray}}
\newcommand{\eea}{\end{eqnarray}}

\newcommand{\eol}{\end{enumerate}\setlength{\itemsep}{-\parsep}}
\newcommand{\eg}{\emph{e.g.}\xspace}
\newcommand{\ie}{\emph{i.e.}\xspace}

\def\argmin{\mathop{\mathrm{arg\,min}}}

\newcommand{\paraspace}{\vspace{0.02in}}
\newcommand{\parab}[1]{\paraspace\noindent{\textbf{#1}}}
\newcommand{\dataem}{\textsc{Movie}}
\newcommand{\dataar}{\textsc{Rating}}
\newcommand{\datagh}{\textsc{GitHub}}
\newcommand{\simadapt}{\textsc{SimAdapt}}
\newcommand{\simfixed}{\textsc{SimFixed}}
\newcommand{\simunif}{\textsc{SimUnif}}
\newcommand{\prsh}{\textsc{PBA}}
\newcommand{\gps}{\textsc{GPS}}

\def\var{\mathop{\mathrm{Var}}}

\def\correl{\mathop{\mathrm{Cor}}}
\def\argmin{\mathop{\mathrm{arg\,min}}}

\newcommand{\WRE}{\textsc{wre}}

\newcommand{\pba}{\textsc{PBA}}
\date{}

\begin{document}
\setlength{\parskip}{0pt plus 1pt minus 1pt}

%title 
\title{Sampling for Approximate Bipartite Network Projection}

\author{
\begin{tabular}{ccc}
Nesreen K. Ahmed&Nick Duffield&Liangzhen Xia
\cr
\normalfont Intel Labs, CA&\normalfont Texas A\&M University&\normalfont 
Texas A\&M University
\cr
\normalfont nesreen.k.ahmed@intel.com&\normalfont duffieldng@tamu.edu&\normalfont xialiangzhen1226@gmail.com
\end{tabular}
}
\maketitle

\begin{abstract}
Bipartite networks manifest as a stream of edges that
represent transactions, e.g., purchases by retail customers. Many machine learning applications employ neighborhood-based measures to characterize the similarity among the nodes, such as
the pairwise number of common neighbors (CN) and related
metrics. While the number of node pairs that share neighbors is
potentially enormous, only a relatively small proportion
of them have many common neighbors. This
motivates finding a weighted sampling approach to preferentially
sample these node pairs. This paper presents a new sampling
algorithm that provides a fixed size unbiased estimate of the
similarity
matrix resulting from a bipartite graph stream projection. The algorithm has two
components. First, it maintains a 
reservoir of sampled bipartite edges with sampling weights that
favor selection of high similarity nodes. Second,
arriving edges generate a stream of \textsl{similarity updates}
based on their adjacency with the current sample. These updates are
aggregated in a second reservoir sample-based stream aggregator to
yield the final unbiased estimate. Experiments on real
world graphs show that a 10\% sample at each stage yields
estimates of high similarity edges with weighted relative
errors of about $10^{-2}$.
\end{abstract}

\section{Introduction}\label{sec:intro}

Networks arise as a natural representation for data, where nodes
represent people/objects and edges represent the relationships among
them. The recent years have witnessed a tremendous amount of research
devoted to the analysis and modeling of complex
networks~\cite{liben2007link}. Bipartite networks are a special class
of networks represented as a graph $G = (U, V, K)$, whose nodes
divide into two sets $U$ and $V$, with edges allowed only between
two nodes that belong to different sets, \ie, $(u,v) \in K$ is an
edge, only if $u \in U$ and $v \in V$. Thus, bipartite networks
represent relationships between two different types of nodes. 
Bipartite networks are a natural model for many systems and
applications. For example, bipartite networks are used to model the relationships between users/customers and the products/services they consume. General examples include collaboration networks in which actors are connected by a common collaboration act (\eg, author-paper, actor-movie) and opinion networks in which users are connected by shared objects (\eg, user-product, user-movie, reader-book). Clearly, a bipartite network manifests as a \emph{stream of edges} representing the transactions between two types of nodes over time, \eg, retail customers purchasing products daily. Moreover, these dynamic bipartite networks are usually large, due to the prolific amount of activity carrying a wealth of useful behavioral data for business analytics.

While the bipartite representation is indeed useful by itself, many
applications focus on analyzing the relationships among a particular
set of nodes~\cite{zhou2007bipartite}. For the convenience of these
applications, a bipartite network is usually compressed by using a
one-mode projection (\ie, projection on one set of the nodes), this is
called \emph{bipartite network projection}. For example, for a one-mode projection on $U$, the projected graph will contain only $U$-nodes and two nodes $u, u' \in U$ are connected if there is at least one common neighbor $v \in V$, such that $(u,v) \in K$ and $(u',v) \in K$. This results in the $U$-projection graph $G_U = (U, K_U, C)$ which is a weighted graph characterized by the set of nodes $U$, and the edges among them $K_U$. The matrix $C = \{C{(u,u')}\}_{U \times U}$ represents the weighted adjacency matrix for the $U$-projection graph, where the weight $C{(u,u')}$ represents the strength of the similarity between the two nodes $u, u' \in  U$.

How to weight the edges has been a key question in one-mode projections and their applications. Several weighting functions were proposed. For example, neighborhood-based methods~\cite{Ning2015,zhou2007bipartite} measure the similarity between two nodes proportional to the overlap of their neighbor sets. Another example in~\cite{fouss2007random} uses random walks to measure the similarity between nodes. Finding similar nodes (\eg, users, objects, items) in a graph is a fundamental problem with applications in recommender systems~\cite{koren2008factorization}, collaborative filtering~\cite{herlocker2004evaluating}, social link prediction~\cite{liben2007link}, text analysis~\cite{salton1993approaches}, among others.

Motivated by these applications, we study the bipartite network projection problem in the \emph{streaming} computational model~\cite{muthu}. Thus, given a bipartite network whose edges arrive as a stream in some arbitrary order, we compute the projection graph (\ie, weighted matrix $C$) as the stream is progressing. In this paper, we focus on the \emph{common neighbors} approach as the weight function. The common neighbors weight is defined for any two nodes $u, u' \in U$, as the size of the intersection of their neighborhood sets $\Gamma(u)$ and $\Gamma(u')$, where $\Gamma(u) \subset V$ is the set of neighbors of $u$. Thus, their projected weight is $C(u,u')= |\Gamma(u) \cap \Gamma(u')|$. It is convenient to think of a bipartite network as a (binary) matrix $A \in \mathbb{R}^{U \times V}$, where the rows represent the node set $U$, and the columns represent the node set $V$. In this case, computing the $U$-projection matrix using common neighbors is equivalent to $C = A A^\intercal$, where $C \in \mathbb{R}^{U \times U}$. In addition, the common neighbors is a fundamental component in many weighing functions (\eg, cosine similarity), such as those used in collaborative filtering.  

The naive solution for this problem is to compute $C = A A^\intercal$
exhaustively with $\mathcal{O}(|U|^2)$ for space and time
complexity. However, this is unfeasible for streaming/large bipartite
networks~\cite{muthu,AhmedTKDD}. Instead, given a streaming bipartite
network (whose edges arrive over time), our goal is to compute a
sample of the projection graph that contains an unbiased estimate of the \emph{largest entries} in the projection matrix $C$.   

\textbf{Contributions}. Our main contribution is a novel single-pass,
adaptive, weighted sampling scheme in fixed storage for
approximate bipartite projection in streaming bipartite networks. Our approach has three steps. First, we maintain a weighted edge sample from the streaming bipartite graph. 
Second, we observe that the number of common neighbors $C(u,u')$ between two vertices $u$ and
$u' \in U$ is equal to the number of wedges $(u,v,u')$ connecting
them, where $(u, v) \in K$ and $(u', v) \in K$ for some $v \in
V$. Thus, each bipartite edge arriving to the sample generates
unbiased estimators of updates to the similarity matrix through the
wedges it creates. Third, a further sample-based aggregation
accumulates estimates of the projection graph in fixed-size storage. 

\section{Framework}\label{sec-framework}

\textbf{Problem Definition and Key Intuition}. Let $G=(U,V,K)$ be a bipartite graph, and $\Gamma(u)=\{v: (u,v)\in K\}$ denote the set of neighbors of $u \in U$. We study the problem of \emph{bipartite network projection} in data streams, where $G$ is compressed by using a one-mode projection. Thus, for a one-mode projection on $U$, the \emph{projected graph} will contain only $U$-nodes and two nodes $u, u' \in U$ are connected if there is at least one common neighbor $v \in V$, such that $(u,v) \in K$ and $(u',v) \in K$. This results in the $U$-projection graph $G_U = (U, K_U, C)$ which is a weighted graph characterized by the set of nodes $U$, and the edges among them $K_U$. The matrix $C = \{C(u,u')\}_{U \times U}$ represents the weighted adjacency matrix for the $U$-projection graph, where the weight $C(u,u')$ represents the strength of the similarity between any two nodes $u, u' \in  U$. In this paper, we propose a novel approximation framework based on sampling to avoid the direct computation of all pairs in $C$.

\begin{mydef}[\textsc{Approximate Bipartite Projection}]\label{def1}  
Given a bipartite network $G = (U, V, K)$ with (binary) adjacency
matrix  $A \in \mathbb{R}^{U \times V}$: find the vertex pair
$(u,u')\in U\times U$ that maximizes $C = A A^\intercal$. More generally, assume a given parameter $k$, find the $k$ vertex pairs $\{(u_1,u'_1), \dots, (u_k,u'_k)\}$ corresponding to the $k$ largest entries in $C$.
\end{mydef}

Note that Definition~\ref{def1} corresponds to finding the pairs with largest number of common neighbors. Intuitively, the number of common neighbors $C(u,u')$ between two vertices $u,u' \in U$, is equivalent to the number of wedges $(u,v,u')$ connecting them, where $(u, v) \in K$ and $(u', v) \in K$ for some $v \in V$. 

\smallskip\noindent
\textbf{Streaming Bipartite Network Projection}. Bipartite networks
are used to model dynamically evolving transactions
represented as a stream of edges between two types of nodes over
time. In the \emph{streaming bipartite graph} model, edges $K$ arrive
in some arbitrary order $\{e_i:i\in[|K|]\}$. Let $K_t=\{e_i:i\in
[t]\}$ denote the first $t$ arriving edges, $G_t=(U_t,V_t,K_t)$ the
bipartite graph induced by the first $t$ arriving edges, and $C_t$ the
corresponding similarity matrix. We aim to estimate the
largest entries of $C_t$ for any $t$. 

\subsection{Adaptive Bipartite Graph Sampling}  
We construct a weighted fixed-size reservoir sample of bipartite edges in which edge weights dynamically adapt to their topological importance (\ie, priority). For a reservoir of size $m$, we admit the first $m$ edges, while for $t>m$, the sample set comprises a subset $\hat K_t \le K_t$ of the first $t$ arriving edges, with fixed size
$|\hat K_t|=m $ for each $t\ge m$. This is achieved by provisionally admitting the arriving edge at each $t>m$ to the reservoir, then discarding one of the resulting $m+1$ edges by the random mechanism that we now describe.

Since edges are assumed unique, each edge $e_i$ is identified with its the arrival order $i\in[|K|]$. 
All sampling outcomes are determined by independent random variables
$\beta_i$, uniformly distributed in $(0,1]$, assigned to each edge $i$ on arrival. Any edge present in the
sample at time $t\ge i$ possess a weight $w_{i,t}$ whose  form
is described in Section~\ref{sec:weights}. The \textbf{priority} of
$i$ at time $t$ is defined as $r_{i,t}=w_{i,t}/\beta_i$. Edge $i$ is
provisionally admitted to the reservoir forming the set $\hat
K'_{i}=\hat K_{i-1}\cup\{i\}$, from which we then discard the edge
with minimal priority, whose value is called the \textbf{threshold}.

Theorem~\ref{thm:unb} below establishes unbiased estimators
of edge counts. Define the \textbf{edge indicator} $S_{i,t}$ taking the value $1$ if
$t\ge i$ and $0$ otherwise. We will construct inverse probability \textbf{edge
  estimators} 
$\hat S_{i,t}=I(i\in\hat K_t)/q_{i,t}$ of $S_{i,t}$ and prove they are 
unbiased. This entails showing that $q_{i,t}=\min\{1, \min_{i\le s\le
  t}w_{i,s}/z_s\}$ is the probability that $i\in K_t$, conditional on
the set $\cZ_{i,t}=\{z_i,\ldots,z_t\}$ of thresholds $z_s=\min_{j\in
  \hat K'_s}r_{j,s}$ since its arrival.

\begin{theorem}\label{thm:unb}
$\hat S_{i,t}$ is an unbiased
estimator of $S_{i,t}$.
\end{theorem}

\begin{proof} Trivially $\hat S_{i,t}=0=S_{i,t}$ for $t<i$.
For $t\ge i$ let $z_{i,t}=\min_{j\in\hat K_t\setminus\{i\}}r_{j,t}$. 
Observe $i\in\hat K_t$ iff $r_{i,s}$ is not the smallest priority
in any $\hat K'_s$ for all $s\in[i,t]$. In other words
\be \{i\in \hat K_t\}
=\cap_{s\in [i,t]}\{\frac{w_{i,s}}{\beta_i}>z_{i,s}\} \\ 
=\{\beta_i<\min_{s\in[i,t]}\frac{w_{i,s}}{z_{i,s}}\}\nonumber
\ee
Thus $\Pr[i\in \hat K_t| \cZ_{i,t}] = \tilde
q_{i,t}:=\min\{1,\min_{s\in[i,t]} w_{i,s}/z_{i,s}\}$. Note that $\tilde
  q_{i,t} = q_{i,t}$ when $i\in \hat K_t$ since then $z_{i,s}=z_s$ for
  all $s\in[i,t]$. Hence when $t\ge i$
\be
\E[\hat S_{i,t}|\cZ_{i,t}] 
=
{\Pr[I(i\in \hat K_t)|\cZ_{i,t}]}/{\tilde q_{i,t}} =1=S_{i,t}
\ee
independent of $\cZ_{i,t}$ and hence $\E[\hat
S_{i,t}]=S_{i,t}$. 
\end{proof}

\noindent 
Let $z^*_t=\max\cZ_{m,t}$. Theorem~\ref{thm:nondec} shows that
$p_{i,t}=\min\{1, \min_{i\le s\le t}\frac{w_{i,s}}{z^*_s}\}$ can be used in place of $q_{i,t}$. This simplifies
computation since: (a) each $p_{i,t}$ uses the same $z^*_t$; (b) updates
of $p_{i,t}$ can be deferred until times $t$ at which $w_{i,t}$ increases.

\begin{theorem}\label{thm:nondec}
If $t\to w_{i,t}$ is non-decreasing for each $i$ then $q_{i,t}=p_{i,t}$
and hence $\hat S_{i,t}=I(i\in \hat K_t)/p_{i,t}$ for all $t\ge i$.
\end{theorem}

\begin{proof}
Let $d_t$ denote the edge discarded during processing arrival
$t$. By assumption,
$i$ is admitted to $\hat K_i$ and since $w_{j,t}$ is non-decreasing in
$t$, $z_i = z_{i,i}=
w_{d_i,i}/\beta_{d_i}\ge w_{d_i,s}/\beta_{d_i}>z_s$ for all $s\in[d_i,i]$ in order that $d_i\in
K_s$ for all $s\in[d_i,i-1]$. Iterating the argument we obtain that
$z_i\ge \cZ_{m,i}$ and hence $z_i=z^*_i$ and $p_{i,i}=q_{i,i}$. The
argument is completed by induction. Assume $p_{i,s}=q_{i,s}$ for
$s>i$. If in addition $z_{t+1}>z_t^*$,  then $z^*_{t+1}=z_{t+1}$ and
hence $p_{i,s+1}=q_{i,s+1}$. If $z_{s+1}\le z^*_s$ then
$z_s=z^*_{s+1}$ and hence $w_{i,s+1}/z_{s+1}\ge w_{i,s+1}/z^*_{s+1}\ge
w_{i,s}/z^*_{s+1}=w_{i,s}/z^*_s$. Thus we replace $z_{s}$ by
$z^*_{s+1}$ in the definition of $q_{i,s+1}$ but use of either leaves
its value unchanged, since by hypothesis both exceed
$q_{i,s}\le w_{i,s}/z^*_{i,s}$.
\end{proof}

\subsection{Edge Sampling Weights}\label{sec:weights} We now specify the weights used for edge
selection. The total similarity of node $u\in U$
is $\cC(u)=\sum_{u'\in U}C(u,u')=\sum_{v\in
  \Gamma(u)} (|\Gamma(v)|-1)$. Thus, the effective contributions of an edge $(u,v)$ to the
total similarities $\cC(u)$ and $\cC(v)$ are $|\Gamma(v)|-1$ and $|\Gamma(u)|-1$
respectively.
This relation indicates that if we wish to sample nodes $u\in U, v\in
V$ with high total similarities $\cC(u)$ and $\cC(v)$ as vertices in the
edge sample, we should sample nodes with high degrees $|\Gamma(u)|$ and
$|\Gamma(v)|$. For adaptive sampling, an edge $e=(u,v)\in \hat K'_t$ has weight
{%\small
\be
w(u,v)=|\hat \Gamma_t(u)| + |\hat \Gamma_t(v)|
\ee
%\normalsize
}
where $\hat \Gamma_t(u),\ \hat \Gamma_t(v)$ are the neighbor sets of $u,\ v$
in the graph $\hat G'_t$ induced by $\hat K'_t$. We also consider a
non-adaptive variant in which edges weights are computed on arrival as above,
but remain fixed thereafter. 

\subsection{Unbiased Estimation of Similarity Weights}
Consider first generating the exact similarity $C_t$ from the
truncated stream $K_t$. Each arriving edge $e_i=(u,v),\ i\le t$ contributes to
$C_t(u, u')$ through wedges $(u,v,u')$ for $v\in \Gamma_t(u)\cap \Gamma_t(u')$. Thus 
to compute $C_t(u,u')$ we count the number of such wedges occurring up to
time $t$, \ie,  
{\small
\bea
C_t(u,u')&=&\sum_{i=1}^t \sum_{v\in \Gamma_i(u)\cap \Gamma_i(u')}
\left( I( u(e_i) = u))S_{(u',v),i-1} \right.\nonumber \\
&&+  \left. I( u(e_i) =
  u'))S_{(u,v),i-1}\right)
\eea
\normalsize}
where $u(e)$ denotes the initial node of edge $e$.
By linearity, we obtain an unbiased estimate $\hat C_t$ of $C_t$ by replacing
each $S_{(u,v),i-1}$ by its unbiased estimate
$\hat S_{(u,v),i-1}$.
Each arriving edge
$e_i=(u,v)$ generates an increment to $\hat C_t(u,u')$ for all edges
$(u',v)$ in $\hat K_i$, the increment size being the
corresponding value of $\hat S_{(u',v),i-1}$, namely,
$1/p_{(u',v),i-1}$. 
%$\hat C_t(v,v')$ for $(v,v')\in V\times V$ is computed similarly.

\begin{algorithm2e}[t]
  \caption{Adaptive Sampling for Bipartite Projection}
  \label{alg5}
\begin{spacing}{0.9}
\fontsize{8}{9}\selectfont
\SetKwFunction{simest}{SimAdapt}
\SetKwFunction{simquery}{SimQuery}
\SetKwFunction{update}{UpdateEdge}
\SetKwFunction{weight}{$w$}
\SetKwFunction{nn}{$N$}
\SetKwFunction{random}{$\beta$}
\SetKwFunction{insertedge}{InsertEdge}
\SetKwFunction{deleteedge}{DeleteEdge}
\SetKwFunction{add}{Add}
\SetKwFunction{open}{Initialize}
\SetKwFunction{close}{Query}
 \SetKwFunction{Aggregate}{Aggregate}
  \SetKwProg{myproc}{Procedure}{}{end}
 \SetKwProg{myfunc}{Function}{}{end}
\SetKwProg{myclass}{Class}{}{end}
\SetKw{mymethod}{Method}
\SetProcNameSty{textsc}
\SetCommentSty{textsl}
\SetFuncSty{textsc}
\SetKwComment{myc}{$\triangleright$\ }{}
\KwIn{Stream of Bipartite Graph Edges in $U\times V$ \; Edge Sample Size $m$; Similiarity
  Sample Size $n$}
\KwOut{Sample Similarity Edges $\hat K_U,\hat K_V$; 
Estimate $\hat C$}
 \myproc{\simest{$m,n$}}{
 $K=\emptyset$; 
  $z^*=0$ \;
  \Aggregate.\open($n$) \;
  \While{(new edge $(u,v)$)}{
    \ForEach{($u'\in
      \Gamma(v)$)}{
      \update{$(u',v),z^*$} \label{alg:l2} \;
      \Aggregate.\add$( (u',u), 1/p(u',v))$ \label{alg:l4} 
      }
    \ForEach{($v'\in
      \Gamma(u)$)}{
      \update{$(u,v'),z^*$} \label{alg:l3} \;
      \Aggregate.\add$((v',v),1/p(u,v'))$ \label{alg:l5} 
}
      \uIf{$|K|< m$}{  
      \insertedge$(u,v)$ \label{alg:l6}\;
    }
     \uElseIf{$\weight(u,v) / \random(u,v) < \min_{(u',v')\in K}\weight(u',v')/\random(u',v')$}{
       $z^* = \max\{z^*,\weight(u,v) / \random(u,v)\}$ \label{alg:l8}
      }
     \Else{
     \insertedge{$(u,v)$} \label{alg:rep1}\;
    $(u^*,v^*)=\argmin_{(u',v')\in K}\weight(u',v')/\random(u',v')$ \;
    $z^*=\max\{z^*,\weight(u^*,v^*)/\random(u^*,v^*)\}$ \label{alg:lmax}\;
    \deleteedge{$(u^*,v^*)$} \label{alg:l9} 
}
  }
}
  \myproc{\update{$(\tilde u,\tilde v), \tilde z$}}{
    \If{$(\tilde z > 0)$} {
      $p(\tilde u,\tilde v)=\min\{p(\tilde u,\tilde v),\weight(\tilde
      u,\tilde v)/\tilde z\}$ \label{alg:updatedef}
    }
  }
 \myproc{\insertedge{$\tilde u, \tilde v$}}{
\ForEach{($u''\in
      \Gamma(\tilde v)$)}{
   \update{$(u'',\tilde v),z^*$} \label{alg:l20} ;
$\weight(u'',\tilde
      v)\mathrel{+}\mathrel{+}$ \label{alg:l10}
}\label{alg:l24}
  \ForEach{($v''\in
     \Gamma(\tilde u)$)}{
\update{$(\tilde u,v''),z^*$} \label{alg:l21}  ;
$\weight(\tilde u,v'')\mathrel{+}\mathrel{+}$ \label{alg:l11}
}
$K=K\cup\{(\tilde u, \tilde v)\};
\weight(\tilde u, \tilde v)=|\Gamma(\tilde u)| + |\Gamma(\tilde v)|; 
p(\tilde u, \tilde v)=1$ \label{alg:l7} %
}
 \myproc{\deleteedge{$\tilde u, \tilde v$}}{
    $K=K\setminus\{(\tilde u, \tilde v)\}$; 
    Delete $p(\tilde u, \tilde
    v)$ 
  }
\hrule
\myproc{\simquery{}}{
$(\hat K_U \cup \hat K_V, \hat C) = $  \Aggregate.\close() 
}
\end{spacing}
\end{algorithm2e}

\subsection{Aggregation of Similarity Updates}
\label{sec:streamsim}
The above construction recasts the problem of reconstituting the sums $\{\hat C_ t(u,u'):
(u,u')\in U\times U\}$ as
the problem of aggregating the stream of key-value pairs
\[
{\small
\left\{\left((u,u'),p^{-1}_{(u',v),i-1}\right): i\in[t],\ e_i=(u,v),\ (u',v)\in \hat
K_i\right\}
\normalsize}
\]
Exact aggregation would entail allocating storage for every
key $(u,u')$ in the stream. Instead, we use weighted sample-based
aggregation to provide 
unbiased estimates of the $\hat C_t(u,u')$ in fixed storage. 
Specific aggregation algorithms with this property include Adaptive Sample \& 
Hold \cite{EV:02short}, Stream VarOpt 
\cite{Cohen:2011:ESS:2079108.2079117} and Priority-Based Aggregation
(\pba) 
\cite{DuffieldCIKM2017}. Each of these schemes is
weighted, inclusion of new items having probability proportional
to the size $1/p_{(u,v),i-1}$ of an update.
Weighted sketch-based methods such as $L_p$
sampling
%\cite{MoWo:SODA2010} 
\cite{6108197,MoWo:SODA2010}
%,Jowhari:Saglam:Tardos:11}. 
could also be used, 
but with space factors that grow
polylogarithmically in the inverse of the bias, they are less able to take
advantage of smoothing from aggregation.

\parab{Estimation Variance.} Inverse probability estimators \cite{HT52}
like those in Theorem~\ref{thm:unb} furnish unbiased variance estimators
computed directly from the
estimators themselves; see \cite{Tille:book}. 
For $\hat S_{i,t}$ this takes the form $\hat V_{i,t} =
\hat S_{i,t}(p_{i,t}^{-1} -1)$, with the unbiasedness property
$\E[\hat V_{i,t}]=\var(\hat S_{i,t})$. These have performed well in
graph stream applications \cite{Ahmed:2017:SMG:3137628.3137651}.
The approach extends to the composite estimators with sample-based
aggregates, using variance
bounds and estimators established for the methods listed
above, combined via the Law of Total Variance.
Due to space limitations we omit the details.

\subsection{Algorithms}\label{sec:algo}
Alg.~\ref{alg5} defines \simadapt\
which implements Adaptive Priority Sampling for bipartite edges,
and generates and aggregates a stream of similarity updates.
It accepts two parameters:
$m$ the reservoir size for streaming bipartite edges, and $n$ the
reservoir size for similarity matrix estimates. Aggregation of 
similarity increments is signified by the class \textsc{Aggregate},
which has three methods: \textsc{Initialize}, which initializes
sampling in a reservoir of size $n$; \textsc{Add}, which aggregates a 
(key,value) update to the similarity estimate, and; \textsc{Query}, which returns
the estimate of the similarity graph at any point in the
stream. 
Each arriving edge generates  similarity updates for each adjacent edge as inverse probabilities (lines ~\ref{alg:l4}, \ref{alg:l5}). 

The bipartite edge sample is maintained in a priority queue $K$ based on
increasing order of edge priority, which for each $(u,v)$ is
computed as the quotient of the edge
weight $w(u,v)$ (the sum of the degrees of $u$ and $v$)
and a permanent random number $\beta(u,v)\in(0,1]$, generated on
demand as a hash of the unique edge
identifier $(u,v)$. The arriving edge is inserted (line \ref{alg:l6}))
if the current occupancy is less than $m$. Otherwise, if its priority is
less than the current minimum, it is discarded and the threshold $z^*$
updated (line~\ref{alg:l8}). If not, the arriving edge replaces the edge of
minimum priority (lines~\ref{alg:rep1}--\ref{alg:l9}).
Edge insertion increments the weights of each adjacent edge
(lines~\ref{alg:l10} and~\ref{alg:l11}). Since $w_{i,t}$ and
$z^*_{t}$ are non-decreasing in $t$, the update of $p_{i,t}$ (\ie, $p_{i,t}=\min\{p_{i,t-1},w_{i,t}/z^*_{t}\}$)
(line~\ref{alg:updatedef}) is deferred until
$w_{i,t}$ increases (lines~\ref{alg:l10},~\ref{alg:l11}) or $p_{i,t}$ is needed for a
similarity update (lines~\ref{alg:l2},~\ref{alg:l3}). 

A variant \simfixed\ uses
(non-adaptive) sampling for bipartite edge sampling with fixed weights. It is
obtained by modifying Algorithm~\ref{alg5} as follows. 
Since weights are not updated, the update and increment steps are
omitted (lines \ref{alg:l2} and \ref{alg:l24}--\ref{alg:l11}).
Edge probabilities are
computed on demand as $p(u,v)=\min\{1,w(u,v)/z^*\}$. We compare with \simunif, a variant
of \simfixed\ with unit weights.

\parab{Data Structure and Time Cost.} 
We implement the priority
queue as a \emph{min-heap}~\cite{Cormen:2001} where the root
position points to the edge with the lowest priority. Access to
the lowest priority edge is $O(1)$. Edge insertions are
$O(\log m)$ worst case.
In \simadapt, each insertion of an edge $(u,v)$ increments the weights
of its neighboring edges. Each weight increment may change 
its edge's priority, requiring its position in the priority queue to be
updated. The worst case cost for heap update is
$O(\log m)$. But since the priority is
incremented, the edge is bubbled down by exchanging with its lowest
priority child if that has lower priority. 

\parab{Space Cost.} 
The space requirement is $O(|\hat V| + |\hat U| + m + n)$, where $|\hat
U|+|\hat V|$ is the number of nodes in the reservoir, with $m$ and $n$
the capacities of the edge and similarity reservoirs.
\begin{comment}
There is a trade-off between space and time, and while we could
limit space to $O(m+n)$, edge update would require a pass
over the reservoir ($O(m)$ worst case). 
\simadapt\ maintains the weight $w(u,v)$ and probability $p(u,v)$
for each bipartite edge $(u,v)$. The priority
$w(u,v)/\beta(u,v)$ is computed on demand.
A hash-based pointer into the heap 
locates the entry of an edge; this
structure is similar to one in
\cite{Metwally:2005:ECF:2131560.2131596}. \simfixed\ does not need
this pointer since it does not update edge weights.
The final sample is $O(n)$, the size of the second stage sampled
aggregate.
\end{comment}

\section{Evaluation}\label{sec-eval}

\begin{table}
{\small
\begin{center}
\begin{tabulary}{0.475\textwidth}{|l|CCCC|CC|}
\hline
&\multicolumn{4}{c|}{Bipartite}&\multicolumn{2}{c|}{Similarity}\\
dataset&$|U|+|V|$&$|K|$&$d_{\mathrm{max}}$&$d_{\textrm{avg}}$&$|K_U|$&$|R_U|$\\
\hline
\dataar&2M&6M&12K&5&204M&203\\
\dataem&62K&3M&33K&90&1.2M&6,797\\
\datagh&122K&440K&4K&7&22.3M&156\\
\hline
\end{tabulary}
\vspace{-2mm}
\caption{Datasets and characteristics. Bipartite graph:
$|U|+|V|$: $\#$nodes, $|K|$: $\#$edges, $d_{\mathrm{max}}$: max. degree,
$d_{\textrm{avg}}$: average degree. Similarity graph: $|K_U|$: $\#$edges
in source similarity, $|R_U|$ $\#$ dense ranks = $\#$distinct weights in source
similarity graph.}
\label{tab:data}
\end{center}
}
\vspace{-3mm}
\end{table}
\textbf{Datasets}. Our evaluations use three datasets comprising bipartite real-world graphs publicly available at
Network Repository~\cite{nr-aaai15}. Basic properties are listed in
Table~\ref{tab:data}.  In the bipartite graph $G=(U,V,K)$, $|U|+|V|$
is the number of nodes in both partitions, $|K|$ is the number of
edges, $d_{\mathrm{max}}$ and $d_{\textrm{avg}}$ are maximum and
average degrees. $|K_U|$ is the number of edges in the source
partition and $|R_U|$ the number of dense ranks, i.e. the number of
distinct similarity values. In \dataar\ (\textsl{rec-amazon-ratings})
an edge indicates a user rated a product; in \dataem\ 
(\textsl{rec-each-movie}) that a user reviewed a
movie, and in \datagh\ (\textsl{rec-github}) that a user is a
member of a project.
The experiments used a 64-bit desktop equipped with an Intel® Core™ i7 Processor with 4 cores running at 3.6 GHz.

\def\uniform{\textsc{simple}}
\def\cnhash{\textsc{CnHash}}

\noindent
\textbf{Accuracy Metrics}. Since applications such as recommendation
systems rank based on similarity, our metrics
focus on accuracy in determining higher similarities that
dominate recommendations with metrics that have been
used in the literature; see e.g., \cite{Gunawardana:2009:SAE:1577069.1755883}.

\textit{Dense Rankings and their Correlation}. We compare estimated and actual rankings of the similarities. We use
\textsl{dense ranking} in which edges with the same similarity have
the same rank, and rank values are consecutive. Dense ranking is
insensitive to permutations of equal similarity edges
 and reduces estimation noise. We use the integer part of the
 estimated similarity to reduce noise.
To assess the linear relationship between the actual and estimated ranks
we use Spearman's rank correlation on top-$k$
actual ranks. For each edge $e=(u,u')$ in the actual
similarity graph, let $r_e$ and $\hat r_e$ denote the dense
ranks of $C(u,u')$ and $\lfloor \hat C(u,u')\rfloor$.
$\correl(k)$ is the top-$k$ rank correlation, i.e., over
pairs $\{(r_e,\hat r_e): e\in K_{U,k}\}$ where $K_{U,k}= \{e\in K_U: r_e \le k\}$.

\textit{Weighted Relative Error.} 
We summarize relative errors 
by weighting by the actual
edge similarity, and for each $k$, we compute the top-$k$
\emph{weighted relative error} $\WRE(k)$ as,
\[
{\small
%\WRE(k)%
%=
\left.{\sum_{(u,u')\in K_{U,k}} \left\vert\hat C(u,u')-C(u,u')\right\vert}\middle/  {\sum_{(u,u')\in
    K_{U,k}} C(u,u')}\right.
\normalsize}
\]

\noindent
\textbf{Baseline Methods}. We compare against two baseline methods. First,
\uniform\ takes a uniform sample of the bipartite edge stream,
and forms an unbiased estimate of $C(u,u')$ by $|\hat \Gamma(u)\cap \hat
\Gamma(u')|/p^2$ where $p$ is the bipartite edge sampling rate. Second, we
compare with sampling-based approach to link prediction in graph streams recently proposed in \cite{7498270},
which investigated several similarity metrics. We
use \cnhash\ to denote its common neighbor (CN) estimate adapted to the bipartite graph
setting. \cnhash\ uses a separate edge sample per node of the full
graph, sampling a fixed maximum reservoir size $L$ per node
using min-hashing to coordinate sampling across different nodes in to
order promote selection of common neighbors. Similarity estimates
are computed across node pairs. Unlike our methods, \cnhash\ does
not offer a fixed bound on the total edge sample size in the streaming
case because neither the number of nodes nor the distribution of
edges is known in advance. We attribute space costs for \cnhash\
using constant space per vertex property of the sketch described in
\cite{7498270}, and map this to an equivalent edge sampling rate
$f_m$, normalizing with the space-per-edge costs of each
method. For a sample aggregate size $n$, we apply
our metrics to the \cnhash\ similarity estimates of the top-$n$ true
similarity edges. 

\noindent
\textbf{Experimental Setup}. We applied \simadapt,
\simfixed\ and \simunif\ to each dataset, using
edge sample reservoir size $m$ a fraction $f_m$ of the total
edges, 
and sample aggregation reservoir size $n$ a
fraction $f_n$ of the edges of the actual similarity graph. 
\dataem\ and \datagh\ used $f_m\in
\{5\%, 10\%, 15\%, 20\%, 25\%, 30\%\}$. The \dataar\ 
achieved the same accuracy with smaller
sampling rates $\{1\%,5\%,10\%\}$. Second stage sampling
fractions were $f_n\in\{5\%, 10\%, 15\%, 100\%\}$, where 100\% 
is exact aggregation.

\section{ Results}\label{sec:eval:res}

\begin{table}
\begin{center}
\small 
\begin{tabular}{|l|l|cc|cc|}
\hline
&&\multicolumn{2}{c|}{\simfixed}&\multicolumn{2}{c|}{\simadapt}\\
dataset&metric&top-100&top-Max&top-100&top-Max\\
\hline
\dataar&\mbox{\WRE}&0.027&0.089&0.012&0.072\\
&\textsl{1-Cor}&\textsl{0.021}&\textsl{0.022}&\textsl{0.009}&\textsl{0.012}\\
\hline
\dataem&\mbox{\WRE}&0.006&0.122&0.002&0.135\\
&\textsl{1-Cor}&\textsl{0.004}&\textsl{0.018}&\textsl{0.001}&\textsl{0.025}\\
\hline
\datagh&\mbox{\WRE}&0.094&0.128&0.064&0.120\\
&\textsl{1-Cor}&\textsl{0.100}&\textsl{0.069}&\textsl{0.046}&\textsl{0.053}\\
\hline
\end{tabular}
\vspace{-2mm}
\caption{Performance of \simadapt, \simfixed\ with $f_m=10\%$
  edge sampling, $f_n=10\%$ \pba. \WRE\ and $1-\correl$.
  Max rank is \{200, 6400, 150\} for  \{\dataar, \dataem,
  \datagh\} }
\label{tab:compare}
\end{center}
\vspace{-3mm}
\end{table}

\parab{Comparison of Proposed Methods.} 
For \simadapt\ and \simfixed, Table~\ref{tab:compare} summarizes the
metrics \WRE\ and $1-\correl$ applied to $\{\dataar, \dataem, \datagh\}$
for both top-100 and maximal dense ranks of $\{200, 6400, 150\}$ 
respectively. The sampling rates are $f_m=10\%$ for bipartite
edges $f_n=10\%$ \pba\ for similarity edges. \simadapt\ performs
noticeably better for the  top-100 dense
ranks, with errors ranging from $0.1\%$ to $6.4\%$ representing
an error reduction of between 32\% and 73\%  relative
to \simfixed. The methods have similar accuracy up to maximal ranks.
%We now breakdown the results further according to
%sampling rates and edge rank.

\begin{figure}[h]
\begin{center}
\begin{tabular}{cc}
\includegraphics[width=1.5in,height=1.7in]{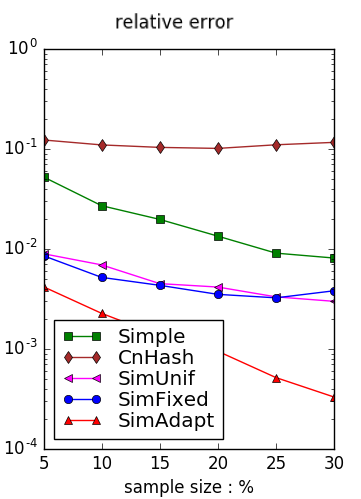}&
\includegraphics[width=1.45in,height=1.7in]{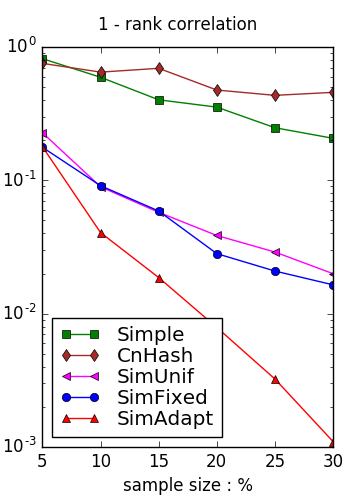}\\
\end{tabular}
\end{center}
\vspace{-5mm}
\caption{Dependence on bipartite edge sample rate $f_m$.
Left: \WRE\ on \dataem. Right: $1-\correl(k)$ on \datagh.
Top-100 dense ranks.}
\label{fig:bakeoff:all}
\end{figure}

\textsl{Accuracy and Bipartite Edge Sample Rate $f_m$.}
Figure~\ref{fig:bakeoff:all} shows metric dependence on edge sample
rate $f_m$ for top-100
dense ranks, using \WRE\ on \dataem\ (right) and $1-\correl$ on
\datagh\ (left). Each figure has
curves for \simadapt, \simfixed\ and \simunif\ (and baseline
methods \cnhash\ and \uniform\ discussed below).  We observe that 
\simadapt\ obtains up to an order of magnitude reduction in both metrics for
$f_m\ge 20\%$.

\begin{figure}[h]
\begin{center}
\begin{tabular}{cc}
\includegraphics[width=1.5in,height=1.6in]{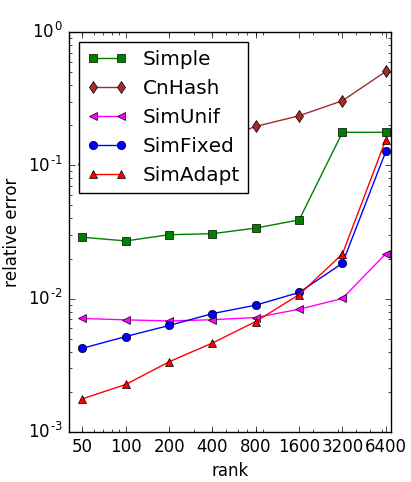}&
\includegraphics[width=1.475in,height=1.6in]{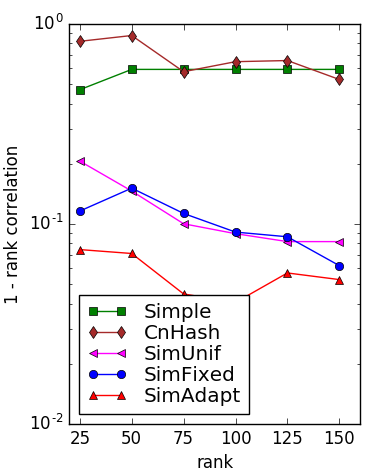}\\
\end{tabular}
\end{center}
\vspace{-5mm}
\caption{Dependence on rank.
Left: \WRE\ on  \dataem. Right: $1-\correl(k)$ on \datagh.
Top-$k$ dense ranks for $k$ up to top-Max.}
\label{fig:bakeoff:rank}
\end{figure}

\textsl{Accuracy and Similarity Rank.}
Figure~\ref{fig:bakeoff:rank} displays the same metric/data
combinations as Figure~\ref{fig:bakeoff:all} with $f_m=10\%$
for top-$k$ ranks
as a function of $k$.
As expected, \simadapt\ is most accurate for lower ranks
that it is designed to sample well, with \WRE\ $0.002$ for \dataem\ at
$k=50$ growing to about $0.2$ at maximum rank
considered. \simunif\ performed slightly better at high
ranks, we believe because it was directing relatively more resources
to high rank edges. 

\textsl{Accuracy and Aggregation Sampling Rate $f_n$.}
In all datasets the \prsh\ second stage had little effect
on accuracy for sampling rates $f_n$  down to about 10\% or less under a wide
variety of parameter settings. 
Figure~\ref{fig:prsh:ranks} shows results for \simfixed\
applied to \dataem\ at fraction $f_m=10\%$ and \prsh\ sampling rates of 5\% and
15\%, specifically \WRE\ and $1-\correl$ for the top-$k$ dense ranks, as a function of $k$. For $k$ up
to several hundred, even 5\% \prsh\ sampling has little or no
effect, while errors roughly double when nearly all ranks are
included. \simunif, and to a lesser extent \simfixed, 
exhibited more noise, even at \textsl{higher}
bipartite sampling rates $f_m$, which we attribute to a greater key diversity
of updates (being less concentrated on high
similarities) competing for space. Indeed, this noise
was absent with exact aggregation.

\begin{figure}
\begin{center}
\begin{tabular}{l}
\includegraphics[height=1.55in,width=3in]{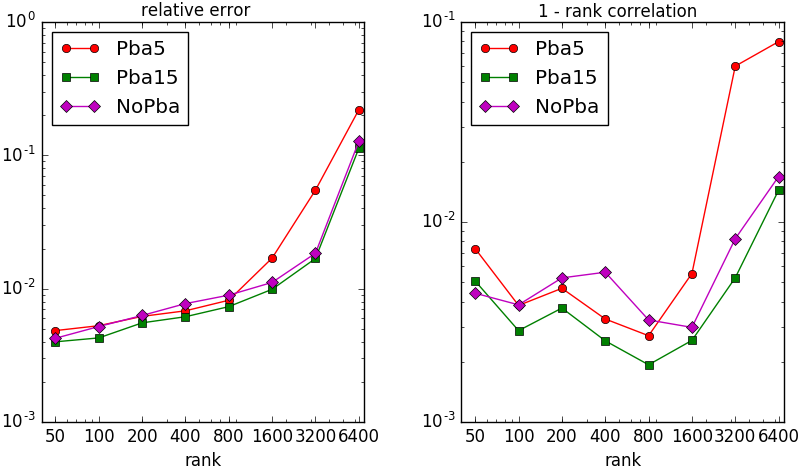}
\end{tabular}
\end{center}
\vspace{-5mm}
\caption{Sample based aggregation w/ \pba\ $f_n=5\% ,15\%$, and none.
Left: $\WRE$. Right: $1-\correl$, as function of top-$k$
ranks. Dataset \dataem\ with \simfixed\ $f_m=10\%$  bipartite edge sampling.}
\label{fig:prsh:ranks}
\vspace{-4mm}
\end{figure}

\parab{Baseline Comparisons.}
Figures~\ref{fig:bakeoff:all} and ~\ref{fig:bakeoff:rank} 
include metric curves for the baseline
methods \uniform\ and \cnhash. \simadapt\ and \simfixed\ typically
performed better than \uniform\ by at least an order of magnitude. In
some experiment with higher edge sample rate $f_m$, \simunif\ was less
accurate than \uniform, we believe due to the noise described above;
%see e.g. Figure~\ref{fig:bakeoff:all}(bottom, left).
Our methods performed noticeably better than \cnhash\ in all cases, 
while \cnhash\ was often no better than \uniform. 
The reasons for this are two-fold. First, in the
streaming context, \cnhash\
does not make maximal use of its constant space per vertex for nodes
whose degree is less than maximum $L$. However,
\textsl{even counting only the stored edges,} \cnhash\ performs worse
than our methods for storage use equivalent to our edge sampling
rate $f_m < 15\%$. This second reason is the interaction of reservoir
design with graph properties. Using shared edge buffer, \simadapt\ and
\simfixed\ devote resources to high adjacency edges associated 
with high similarity in a sparse graph. Edges incident at high degree nodes are more likely
to acquire future adjacencies. 

Noise reduction was
employed for similarity estimates comprising
a small number of updates. These exhibit noise from 
inverse probability estimators without the benefit of
smoothing. We maintained an update count per edge and 
filtered estimates with count below a
threshold.
Most benefit
was obtained by filtering estimates of update count below
10; this was used in all experiments
reported above. Figure~\ref{fig:filter2}
compares the effects of no filtering
(top) with filtering at threshold 10 (right) applied  to \simadapt\
with $f_m=$ 10\% edge sample, for the
approximately 1,000 similarity edges in the top 100
estimated dense ranks. The left column shows actual and forecast weights.
Without filtering, noise in the estimated similarity curve
is due to a few
edges whose estimated similarity greatly exceeds
the actual similarity due to estimation noise, These are largely absent after
filtering. The right column shows a scatter of (actual, forecast)
ranks. Observe the cluster of edges with high actual rank (i.e. lower actual
weight) and overestimated weight present with no filtering, that are removed by filtering.

\begin{figure}
\begin{center}
\begin{tabular}{cc}
\includegraphics[width=1.5in,height=1.55in]{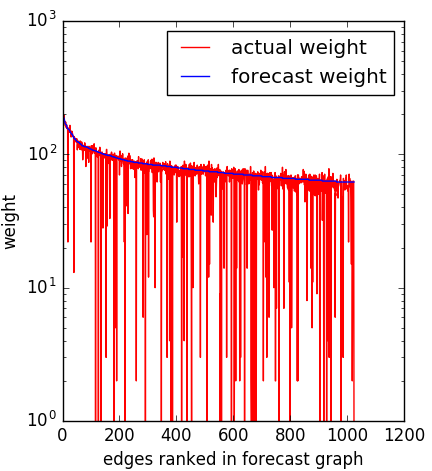}&
\includegraphics[width=1.5in,height=1.55in]{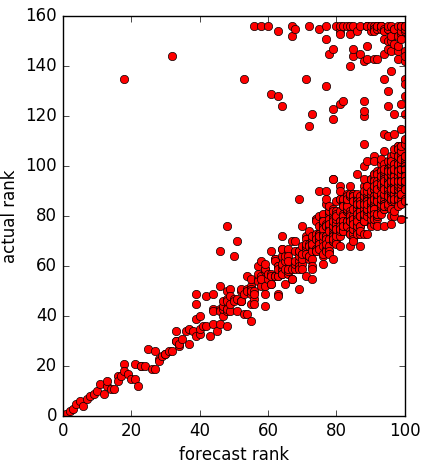}\\
\includegraphics[width=1.5in,height=1.55in]{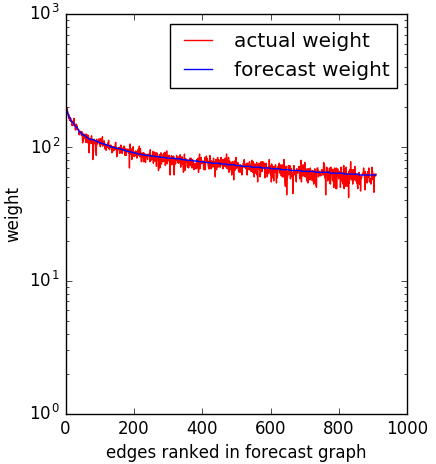}&
\includegraphics[width=1.5in,height=1.55in]{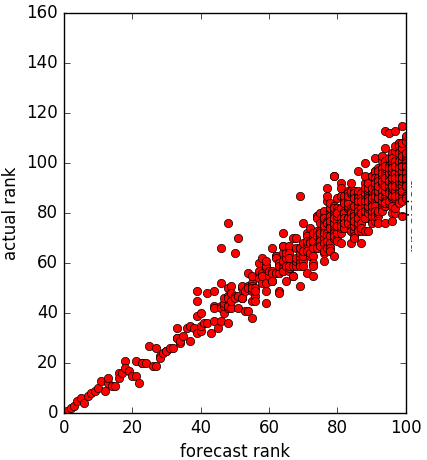}\\
\end{tabular}
\end{center}
\vspace{-5mm}
\caption{Noise and Filtering. \datagh\ $f_m=10\%$. Top-100
  dense ranked edges. Top: No
  filtering. Bottom: filter threshold 10. Left: forecast and actual
  weights. Right: scatter of (forecast, actual) ranks.}
\label{fig:filter2}
\end{figure}

\section{Related Work}\label{sec-related}

A number of problems specific to bipartite graphs have recently
attracted attention in the streaming or semi-streaming
context. The classic problem of bipartite matching has been considered
for semi-streaming  \cite{Eggert:2012:BMS:2150816.2150835,Kliemann2011} and
streaming \cite{Goel:2012:CSC:2095116.2095157} data.
Identifying top-k queries in graphs streams has been studied in
\cite{Pan:2012:CTQ:2396761.2398717}.  The
Adaptive Graph Priority Sampling of this paper builds on the graph priority sampling framework 
\gps\ in \cite{Ahmed:2017:SMG:3137628.3137651} while the
second sample aggregation method appears in
\cite{DuffieldCIKM2017}. Graph stream sampling for subgraph counting
is addressed in
\cite{Ahmed:2017:SMG:3137628.3137651,Jha:2015:SSA:2737800.2700395,Stefani:2017:TCL:3119906.3059194,ZB2017,Ahmed-gSH}
amongst others; see \cite{AhmedTKDD} for a review. \cite{7498270}
is closer to our present work in that it provides a sample-based
estimate of the CN count, albeit not specialized to the bipartite
context. We make a detailed comparison of design and performance of 
\cite{7498270} with our proposed approach in in Section~\ref{sec:eval:res}.

\section{Summary and Conclusion}\label{sec-conclusion}

This paper has proposed a sample-based estimator of the
similarity (or projection graph) induced by a bipartite edge
stream, i.e., the weighted graph whose edge weights or similarities are the numbers
of common neighbors of its endpoint nodes. 
The statistical properties of real-world bipartite graphs
provide an opportunity for weighted sampling that devotes resources to
nodes with high similarity edges in the projected graph. 
Our proposed algorithm provides unbiased estimates of
similarity graph edges in fixed storage without prior knowledge of the
graph edge stream. With a relatively small sample of
bipartite and the similarity graph edges (10\% in each case), and with
the enhancement of count based filtering of similarity edges at
threshold 10, the sampled similarity edge set reproduces the actual
similarities of sampled edges with errors of about $10^{-2}$ for top-100 dense estimate ranked edges, rising to an error of
about $10^{-1}$ when most estimated edges are considered.  Indeed,
for the parameters used, the rank distribution
of the sampled similarity graph is very similar to that
of the actual graph for all but the highest ranks.

\bibliographystyle{named}      
\bibliography{paper,cycle,ahmed,bibmaster,psh,bipartite}

\begin{thebibliography}{}

\bibitem[\protect\citeauthoryear{Ahmed \bgroup \em et al.\egroup
  }{2014a}]{Ahmed-gSH}
N.~K. Ahmed, N.~Duffield, J.~Neville, and R.~Kompella.
\newblock Graph sample and hold: A framework for big-graph analytics.
\newblock In {\em SIGKDD}, 2014.

\bibitem[\protect\citeauthoryear{Ahmed \bgroup \em et al.\egroup
  }{2014b}]{AhmedTKDD}
N.~K. Ahmed, J.~Neville, and R.~Kompella.
\newblock Network sampling: From static to streaming graphs.
\newblock {\em In TKDD}, 8(2):1--56, 2014.

\bibitem[\protect\citeauthoryear{Ahmed \bgroup \em et al.\egroup
  }{2017}]{Ahmed:2017:SMG:3137628.3137651}
Nesreen~K. Ahmed, Nick Duffield, Theodore~L. Willke, and Ryan~A. Rossi.
\newblock On sampling from massive graph streams.
\newblock {\em Proc. VLDB}, 10(11):1430--1441, August 2017.

\bibitem[\protect\citeauthoryear{Andoni \bgroup \em et al.\egroup
  }{2011}]{6108197}
A.~Andoni, R.~Krauthgamer, and K.~Onak.
\newblock Streaming algorithms via precision sampling.
\newblock In {\em 2011 IEEE 52nd Annual Symposium on Foundations of Computer
  Science}, pages 363--372, Oct 2011.

\bibitem[\protect\citeauthoryear{Cohen \bgroup \em et al.\egroup
  }{2011}]{Cohen:2011:ESS:2079108.2079117}
Edith Cohen, Nick Duffield, Haim Kaplan, Carsten Lund, and Mikkel Thorup.
\newblock Efficient stream sampling for variance-optimal estimation of subset
  sums.
\newblock {\em SIAM J. Comput.}, 40(5):1402--1431, September 2011.

\bibitem[\protect\citeauthoryear{Cormen \bgroup \em et al.\egroup
  }{2001}]{Cormen:2001}
Thomas~H. Cormen, Clifford Stein, Ronald~L. Rivest, and Charles~E. Leiserson.
\newblock {\em Introduction to Algorithms}.
\newblock 2nd edition, 2001.

\bibitem[\protect\citeauthoryear{Duffield \bgroup \em et al.\egroup
  }{2017}]{DuffieldCIKM2017}
Nick~G. Duffield, Yunhong Xu, Liangzhen Xia, Nesreen~K. Ahmed, and Minlan Yu.
\newblock Stream aggregation through order sampling.
\newblock In {\em CIKM}, 2017.

\bibitem[\protect\citeauthoryear{Eggert \bgroup \em et al.\egroup
  }{2012}]{Eggert:2012:BMS:2150816.2150835}
Sebastian Eggert, Lasse Kliemann, Peter Munstermann, and Anand Srivastav.
\newblock Bipartite matching in the semi-streaming model.
\newblock {\em Algorithmica}, 63:490--508, 2012.

\bibitem[\protect\citeauthoryear{Estan and Varghese}{2002}]{EV:02short}
C.~Estan and G.~Varghese.
\newblock New directions in traffic measurement and accounting.
\newblock In {\em Proc. ACM SIGCOMM~'2002}, Pittsburgh, PA, 2002.

\bibitem[\protect\citeauthoryear{Fouss \bgroup \em et al.\egroup
  }{2007}]{fouss2007random}
Francois Fouss, Alain Pirotte, Jean-Michel Renders, and Marco Saerens.
\newblock Random-walk computation of similarities between nodes of a graph with
  application to collaborative recommendation.
\newblock {\em IEEE Transactions on knowledge and data engineering},
  19(3):355--369, 2007.

\bibitem[\protect\citeauthoryear{Goel \bgroup \em et al.\egroup
  }{2012}]{Goel:2012:CSC:2095116.2095157}
Ashish Goel, Michael Kapralov, and Sanjeev Khanna.
\newblock On the communication and streaming complexity of maximum bipartite
  matching.
\newblock In {\em Proc. SODA '12}, pages 468--485, Philadelphia, PA, USA, 2012.

\bibitem[\protect\citeauthoryear{Gunawardana and
  Shani}{2009}]{Gunawardana:2009:SAE:1577069.1755883}
Asela Gunawardana and Guy Shani.
\newblock A survey of accuracy evaluation metrics of recommendation tasks.
\newblock {\em J. Mach. Learn. Res.}, 10, 2009.

\bibitem[\protect\citeauthoryear{Herlocker \bgroup \em et al.\egroup
  }{2004}]{herlocker2004evaluating}
Jonathan~L Herlocker, Joseph~A Konstan, Loren~G Terveen, and John~T Riedl.
\newblock Evaluating collaborative filtering recommender systems.
\newblock {\em ACM Transactions on Information Systems (TOIS)}, 22(1):5--53,
  2004.

\bibitem[\protect\citeauthoryear{Horvitz and Thompson}{1952}]{HT52}
D.~G. Horvitz and D.~J. Thompson.
\newblock A generalization of sampling without replacement from a finite
  universe.
\newblock {\em J. of the American Stat. Assoc.}, 47(260):663--685, 1952.

\bibitem[\protect\citeauthoryear{Jha \bgroup \em et al.\egroup
  }{2015}]{Jha:2015:SSA:2737800.2700395}
Madhav Jha, C.~Seshadhri, and Ali Pinar.
\newblock A space-efficient streaming algorithm for estimating transitivity and
  triangle counts using the birthday paradox.
\newblock {\em ACM Trans. Knowl. Discov. Data}, 9(3):15:1--15:21, 2015.

\bibitem[\protect\citeauthoryear{Kliemann}{2011}]{Kliemann2011}
Lasse Kliemann.
\newblock {\em Matching in Bipartite Graph Streams in a Small Number of
  Passes}, pages 254--266.
\newblock Springer, Berlin, Heidelberg, 2011.

\bibitem[\protect\citeauthoryear{Koren}{2008}]{koren2008factorization}
Yehuda Koren.
\newblock Factorization meets the neighborhood: a multifaceted collaborative
  filtering model.
\newblock In {\em Proceedings of the 14th ACM SIGKDD international conference
  on Knowledge discovery and data mining}, pages 426--434. ACM, 2008.

\bibitem[\protect\citeauthoryear{Liben-Nowell and
  Kleinberg}{2007}]{liben2007link}
David Liben-Nowell and Jon Kleinberg.
\newblock The link-prediction problem for social networks.
\newblock {\em journal of the Association for Information Science and
  Technology}, 58(7):1019--1031, 2007.

\bibitem[\protect\citeauthoryear{Monemizadeh and
  Woodruff}{2010}]{MoWo:SODA2010}
M.~Monemizadeh and D.~P. Woodruff.
\newblock 1-pass relative-error l$_{\mbox{p}}$-sampling with applications.
\newblock In {\em Proc. 21st ACM-SIAM Symposium on Discrete Algorithms}.
  ACM-SIAM, 2010.

\bibitem[\protect\citeauthoryear{Muthukrishnan}{2005}]{muthu}
S.~Muthukrishnan.
\newblock {\em Data streams: Algorithms and applications}.
\newblock Now Publishers Inc, 2005.

\bibitem[\protect\citeauthoryear{Ning \bgroup \em et al.\egroup
  }{2015}]{Ning2015}
Xia Ning, Christian Desrosiers, and George Karypis.
\newblock {\em A Comprehensive Survey of Neighborhood-Based Recommendation
  Methods}, pages 37--76.
\newblock Springer US, Boston, MA, 2015.

\bibitem[\protect\citeauthoryear{Pan and
  Zhu}{2012}]{Pan:2012:CTQ:2396761.2398717}
Shirui Pan and Xingquan Zhu.
\newblock Continuous top-k query for graph streams.
\newblock In {\em Proc. CIKM '12}, pages 2659--2662, New York, NY, USA, 2012.

\bibitem[\protect\citeauthoryear{Rossi and Ahmed}{2015}]{nr-aaai15}
Ryan~A. Rossi and Nesreen~K. Ahmed.
\newblock The network data repository with interactive graph analytics and
  visualization.
\newblock In {\em AAAI}, 2015.
\newblock \url{http://networkrepository.com}.

\bibitem[\protect\citeauthoryear{Salton \bgroup \em et al.\egroup
  }{1993}]{salton1993approaches}
Gerard Salton, James Allan, and Chris Buckley.
\newblock Approaches to passage retrieval in full text information systems.
\newblock In {\em ACM SIGIR 1993}, pages 49--58. ACM, 1993.

\bibitem[\protect\citeauthoryear{Stefani \bgroup \em et al.\egroup
  }{2017}]{Stefani:2017:TCL:3119906.3059194}
Lorenzo~De Stefani, Alessandro Epasto, Matteo Riondato, and Eli Upfal.
\newblock Tri\`{E}st: Counting local and global triangles in fully dynamic
  streams with fixed memory size.
\newblock {\em ACM TKDD}, 11(4):43:1--43:50, 2017.

\bibitem[\protect\citeauthoryear{Till\'e}{2006}]{Tille:book}
Y.~Till\'e.
\newblock {\em Sampling Algorithms}.
\newblock Springer-Verlag, 2006.

\bibitem[\protect\citeauthoryear{Zakrzewska and Bader}{2017}]{ZB2017}
Anita Zakrzewska and David~A Bader.
\newblock Streaming graph sampling with size restrictions.
\newblock In {\em IEEE/ACM International Conference on Advances in Social
  Networks Analysis and Mining}, 2017.

\bibitem[\protect\citeauthoryear{Zhao \bgroup \em et al.\egroup
  }{2016}]{7498270}
P.~Zhao, C.~Aggarwal, and G.~He.
\newblock Link prediction in graph streams.
\newblock In {\em Proc. ICDE '16}, pages 553--564, May 2016.

\bibitem[\protect\citeauthoryear{Zhou \bgroup \em et al.\egroup
  }{2007}]{zhou2007bipartite}
Tao Zhou, Jie Ren, Mat{\'u}{\v{s}} Medo, and Yi-Cheng Zhang.
\newblock Bipartite network projection and personal recommendation.
\newblock {\em Physical Review E}, 76(4):046115, 2007.

\end{thebibliography}

\end{document}